\documentclass[reqno]{amsart}
\usepackage{amssymb,graphics,graphicx,bbm, float}

\def\be{\begin{equation}}
\def\ee{\end{equation}}

\newcommand{\vol}{\mathop{\mathrm{Vol}}}
\newcommand{\sign}{\mathop{\mathrm{sign}}}

\numberwithin{equation}{section}
\newtheorem{theorem}{Theorem}[section]
\newtheorem{proposition}[theorem]{Proposition}
\newtheorem{lemma}[theorem]{Lemma}
\newtheorem{corollary}[theorem]{Corollary}

\newtheorem{definition}{Definition}
\newtheorem{remark}{Remark}

\begin{document}

\title{A Solution to the Combinatorial Puzzle of Mayer's Virial Expansion}

\author[S. J. Tate]{Stephen James Tate}
\address{Ruhr Universit\"{a}t Bochum, Universit\"{a}tsstr. 150, 44801 Bochum}
\email{sjtate1988@gmail.com}



\begin{abstract}
Mayer's second theorem in the context of a classical gas model allows us to write the coefficients of the virial expansion of pressure in terms of weighted two-connected graphs. Labelle, Leroux and Ducharme studied the graph weights arising from the one-dimensional hardcore gas model and noticed that the sum of these weights over all two-connected graphs with $n$ vertices is $-n(n-2)!$. This paper addresses the question of achieving a purely combinatorial proof of this observation.
\end{abstract}

\thanks{Work supported by EPSRC grant EP/G056390/1 and SFB TR12}
\thanks{\copyright{} 2014 by the author. This paper may be reproduced, in its entirety, for non-commercial purposes.}

\maketitle

\section{Introduction}
This paper considers (multivariate) generating functions of the form:
\be \sum\limits_{n=1}^{\infty} \frac{z^n}{n!} \sum\limits_{G \in \mathcal{H}[n]} \tilde{w}(G) s^{e(G)} , \label{eq:genfn} \ee
where $\mathcal{H}$ indicates a subclass of graphs. $e(G)$ is the number of edges a graph $G$. $[n]$ indicates that the graph has vertex set $\{1, \cdots, n\}=:[n]$ and $\tilde{w}$ is a specified positive graph weight. The exponents of the variables $z$ and $s$ indicate the size of the vertex set, respectively, the number of edges.

When evaluating \eqref{eq:genfn} at $s=-1$, there are some remarkable cancellations, leading, in some cases, to simple formul\ae{} for the coefficients. This paper gives combinatorial explanations for the class of two-connected graphs in particular. Two-connected graphs are those graphs for which we can remove any vertex and its incident edges and the resulting graph remains connected. 

There are four cases of \eqref{eq:genfn} used in Mayer's theory of cluster and virial expansions, depending on the class of graphs considered and the weights. The sum is either over connected graphs, denoted by $\mathcal{C}$, or two-connected graphs, denoted by $\mathcal{B}$. The weights are either those for a discrete hard core gas, often referred to as the one-particle hard core gas, or for a continuum one-dimensional hard core gas, also named the Tonks gas. For the discrete gas, the goal is to count the number of graphs; for the continuum model, the coefficients are given by the volume of a polytope associated with the graph $G$. We write a graph $G$ as the ordered pair of its vertex set and edge set as $(V(G),E(G))$. 

We define the polytope corresponding to the graph $G$ as:
\be \Pi_G :=\{(\mathbf{x})_{[2,n]}\in \mathbb{R}^{n-1} \vert \, |\mathbf{x}_i-\mathbf{x}_j|<1 \; \forall\{i,j\} \in E(G) \} \label{eq:graphpolytope} , \ee
with $x_1=0$. We use the notation $(\mathbf{x})_{[2,n]}:=(x_2, \cdots, x_n)$.

Mayer, in \cite{MM40}, established important connections between weighted graph generating functions and expansions in statistical mechanics. These connections are also presented in the framework of combinatorial species of structure in the work of Ducharme, Labelle and Leroux \cite{DLL07, L04}, Leroux and Kaouche \cite{KL08} and Faris \cite{F10}. 

The results of Mayer are that the weighted sum over connected graphs gives the pressure as a function of activity and the weighted sum over two-connected graphs is related to the virial expansion of pressure expanded in terms of density. The two formul\ae{} are:
\begin{align} \beta P(z)&=\sum\limits_{n=1}^{\infty} \frac{z^n}{n!} \sum\limits_{G \in \mathcal{C}[n]} w(G) \\
\beta P(\rho) &= \rho - \sum\limits_{n=2}^{\infty}(n-1) \frac{\rho^n}{n!} \sum\limits_{G \in \mathcal{B}[n]}w(G) , \end{align}
where $w(G)$ is the graph weight specified by the particular model.

The answers for the four cases are given by the formul\ae{} for the connected graph discrete case:
\be \sum\limits_{G \in \mathcal{C}[n]}(-1)^{e(G)}= (-1)^{n-1}(n-1)!  \ee
and for the connected graph continuum case:
\be \sum\limits_{G \in \mathcal{C}[n]}(-1)^{e(G)}\vol(\Pi_G) = (-1)^{n-1}n^{n-1} . \ee
There are also formul\ae{} for the two-connected discrete case:
\be \sum\limits_{G \in \mathcal{B}[n]}(-1)^{e(G)}=-(n-2)! \label{eq:virial1PHC} \ee
and the two-connected graph continuum case:
\be \sum\limits_{G \in \mathcal{B}[n]}(-1)^{e(G)}\vol(\Pi_G)=-n(n-2)! \label{eq:virialTG} .\ee
For the discrete cases the results are straightforward computations. For the continuum case, derivations are given in \cite{DLL07}. The statistical mechanical background is explained in full detail in Section \ref{sec:models}.

It is tempting to try and find a simple combinatorial interpretation that explains the cancellations in a direct way. This was posed as a challenge in the paper of Ducharme, Labelle and Leroux \cite{DLL07}. In the connected graph cases, this was done by Bernardi \cite{B08}. The approach was to use an involution that exhibits the result of the almost perfect cancellation as a contribution from the fixed points of the involution. The fixed points were identified as increasing trees in the discrete case and rooted trees in the continuum case. The purpose of this paper is to present similar derivations for the two-connected graph cases. As always this is considerably more complicated.  
 
The concept of using an involution to understand the cancellations is natural. Recall the formula that, for any finite non-empty set $X$, we have: 
\be \sum\limits_{S \subseteq X}(-1)^{|S|}=0 . \ee
In order to prove this, we show we have the same number of sets with even cardinality as we do of odd cardinality. One approach is to pair sets differing by one element. This pairing idea is captured by the involution. In this example, the involution is defined by first fixing a singleton subset of $X$, say $\{i\}$, and taking the symmetric difference $\Psi: S \mapsto S \Delta \{i\}$. 

If we consider a fixed vertex set $[n]$ for a graph, then a graph $G$ is determined precisely by its edge set $E(G)$, which are subsets of the collection of unordered pairs in $[n]$, denoted $[n]^{(2)}$. We can also use this symmetric difference operation on the edge set for graphs. An important complication is that we consider particular subsets for which taking the symmetric difference with a fixed edge will not suffice, since the removal or addition of the edge may take us outside of the prescribed collection of subsets. We need to find an efficient way of choosing an edge based on the graph we are considering so that we obtain a pairing that will not take us outside of the prescribed collection.

In section \ref{sec:results}, we present the combinatorial structures that give the interpretations of the cancellations in the two-connected case. Sections \ref{sec:1PHCinv} and \ref{sec:CHCGinv} give the proofs of the one particle hard core and the Tonks gas case respectively. In the latter, the decomposition of polytopes into unimodular simplices attributed to Lass is given so that it may be proved as an extension of the previous case. We provide an interpretation why $2n-3$ should appear as the number of edges in section \ref{sec:structure}.

From the perspective of statistical mechanics, the motivation for understanding such cancellations is to be able to adapt the understanding to models where more complicated weights are used. Indeed, the key idea is to emulate what is done for the connections between connected graphs and trees and understand how to modify these in this context. 

The first parallel to draw is that the involution of Bernardi fits within a general concept of externally and internally active elements of a set with a matroid structure as given by Bj\"{o}rner and Sokal \cite{B92, S05}. The idea to emphasise here is that this allows the set of connected graphs to be partitioned into subsets, indexed by trees. When we consider graphs with the partial order defined by bond inclusion, the blocks in this partition are Boolean. That is, each set has a tree $\tau$ as minimal graph and a corresponding maximal graph $R(\tau)$, all graphs with edge set $E$ such that $E(\tau) \subset E \subset E(R(\tau))$ are included in the set in the partition. This form of a partition lends itself well to performing estimates on the cluster coefficients. This was actually realised earlier by Penrose \cite{P67} in the specific case of connected graphs. Understanding this partition into Boolean subsets also gives rise to an alternative involution. It is intriguing to realise that the general construction does not include the Penrose construction as a subcase. These ideas are addressed in section \ref{sec:extend}.

This combinatorial understanding is also closely linked to the tree-graph identities of Brydges Battle and Federbush \cite{BF78, B84, B84a, BF84b}, for which a symmetric version is provided by Abdesselam and Rivasseau \cite{AR94} and a matroid generalisation by Faris \cite{F12}. These identities allow estimations to be made on these coefficients, since we may express the sum over connected graphs as a sum over trees with modified weights. A greater goal is to extend these to partially ordered sets where a matroid structure may not be present. 

Interest in providing such bounds on the virial expansion coefficients has recently been renewed with the papers by Pulvirenti and Tsagkarogiannis \cite{PT11} and Morais and Procacci \cite{MP13}, which use the Canonical Ensemble as a method of achieving bounds. The paper by Jansen \cite{J11b} suggests that at high temperatures the radius of convergence should be improved: actual improvements on the bounds of Lebowitz and Penrose \cite{LP64} have been proposed recently \cite{T13}. 

\section{The Two Models from Statistical Mechanics}
\label{sec:models}
In a classical gas system of $n$ indistinguishable interacting particles in a vessel $\Lambda \subset \mathbb{R}^d$ with only two-body interactions and no external potential, we may write the Hamiltonian as:
\be H(\mathbf{p},\mathbf{q})=\sum\limits_{i=1}^n \frac{p_i^2}{2m}+\sum\limits_{1 \leq i<j \leq n}\varphi(q_i,q_j) , \ee
where $\mathbf{q}$ represents the generalised coordinates and $\mathbf{p}$ the conjugate momenta. The canonical partition function of the gas model is:
\be Z(\Lambda,\beta,n)=\frac{1}{n!}\prod\limits_{i=1}^n \left( \int_{\Lambda}\! \mathrm{d}^dq_i \int_{\mathbb{R}^d}\! \mathrm{d}^dp_i\right) \exp(- \beta H) . \ee
Integrating out the Gaussian integrals for the momenta, we obtain a factor $\frac{1}{\lambda^n}$, where $\lambda$ is the thermal wavelength. The partition function is therefore:
\be Z(\Lambda, \beta,n)=\frac{1}{n!\lambda^n} \prod\limits_{i=1}^n \left(\int_{\Lambda}\! \mathrm{d}^dq_i\right) \prod\limits_{1 \leq i<j \leq n}\exp(-\beta \varphi(q_i,q_j)) \label{eq:canonicalpfinvol} . \ee
The Mayer trick \cite{MM40}, allows us to rewrite the canonical partition function in terms of weighted graphs. The first stage is to define the Mayer $f$-function:
\be f(q_i,q_j):=\exp(-\beta \varphi(q_i,q_j))-1 .\ee
We realise that the product of exponentials in \eqref{eq:canonicalpfinvol} may be rewritten as:
\be \prod\limits_{1 \leq i < j \leq n} \exp(-\beta \varphi(q_i,q_j))=\prod\limits_{1 \leq i<j \leq n}(1+f(q_i,q_j))=\sum\limits_{G \in \mathcal{G}[n]}\prod\limits_{(i,j) \in E(G)}f(q_i,q_j) , \ee
where $\mathcal{G}[n]$ is the set of simple graphs (no multiple edges or loops) on $n$ points. We write a graph $G=(E(G),V(G))$, where $E(G) \subset [n]^{(2)}$ is the edge set and $V(G)=[n]$ is the vertex set. This motivates the graph weight:
\be W(G)=\prod\limits_{i=1}^n \left(\int_{\Lambda}\! \mathrm{d}^dq_i\right) \prod\limits_{(k,l) \in E(G)}f(q_k,q_l). \ee
We can therefore write the partition function as:
\be Z(\Lambda, \beta, n)=\frac{1}{n!\lambda^n}\sum\limits_{G \in \mathcal{G}[n]}W(G). \ee
In order to obtain the grand canonical partition function we sum:
\be \Xi(\Lambda, \beta, z)=\sum\limits_{n=0}^{\infty} z^n\lambda^n Z(\lambda, \beta, n), \ee
where $z=e^{\beta \mu}$ the activity and $\mu$ is the chemical potential. In terms of graphs, we write this as:
\be \Xi(\Lambda, \beta, z)=\sum\limits_{n=0}^{\infty}\frac{z^n}{n!}\sum\limits_{G \in \mathcal{G}[n]}W(G) =: \mathcal{G}_W(z). \ee
The pressure is defined to be:
\be \beta P= \lim_{|\Lambda| \uparrow \infty}\frac{1}{|\Lambda|}\log \Xi(\Lambda, \beta, z) .\ee
If we define the new weight $w(G)=\lim_{|\Lambda| \uparrow \infty} \frac{1}{|\Lambda|}W(G)$, then the pressure function can be written in terms of connected graphs:
\be \beta P = \mathcal{C}_w(z)=\sum\limits_{n=1}^{\infty} \frac{z^n}{n!}\sum\limits_{G \in \mathcal{C}[n]}w(G) \label{eq:pressureconnectedinv} .\ee
This is the content of Mayer's First Theorem \cite{MM40} and is explained in the paper \cite{DLL07}.
The density $\rho$ is:
\be \rho= z\frac{\partial}{\partial z} \beta P = \mathcal{C}^{\bullet}_w(z) , \ee
where $\mathcal{C}^{\bullet}$ denotes a rooted connected graph. From Mayer's Second Theorem \cite{MM40} or by the Dissymmetry Theorem \cite{DLL07}, we are able to obtain a series expansion for pressure in terms of density, in which the coefficients are, up to a prefactor, the $w$-weighted two-connected graphs.
\be \beta P = \rho - \sum\limits_{n=2}^{\infty}\frac{(n-1)\rho^n}{n!} \sum\limits_{G \in \mathcal{B}[n]}w(G) \label{eq:pressure2connectedinv}.\ee
One may also consult the book by McCoy \cite{M10} for an explanation of the derivation of these two theorems. 
\subsection{One Particle Hard Core Gas}
The potential for a one-particle hard core gas is:
\be \varphi(q_i,q_j)=\infty , \ee
so that $\exp(-\beta \varphi(q_i,q_j))=0$ and $f(q_i,q_j)=-1$.
The grand canonical partition function is: 
\be \Xi(z)=1+z . \ee
The statistical mechanical relationships give pressure and density as:
\begin{align} \beta P &= \log(1+z) \label{eq:1PHCpressure}\\ \rho&=\frac{z}{1+z}  \label{eq:1PHCdensity}. \end{align}
We may invert \eqref{eq:1PHCdensity}, to obtain:
\be z = \frac{\rho}{1-\rho} \ee
and substitute for $z$ in \eqref{eq:1PHCpressure}, to obtain:
\be \beta P = - \log(1-\rho) . \ee
The two series expansions derived from statistical mechanics are:
\begin{align} \beta P &= \sum\limits_{n=1}^{\infty} \frac{(-1)^{n-1}z^n}{n} \\
\beta P &=\sum\limits_{n=1}^{\infty}\frac{\rho^n}{n} .\end{align}
If we compare these two power series with \eqref{eq:pressureconnectedinv} and \eqref{eq:pressure2connectedinv} respectively, using the graph weight $w(G)=(-1)^{e(G)}$, where $e(G)$ is the number of edges in graph $G$, we obtain:
\begin{align} \sum\limits_{G \in \mathcal{C}[n]}(-1)^{e(G)}&=(-1)^{n-1}(n-1)! \\
\sum\limits_{G \in \mathcal{B}[n]}(-1)^{e(G)}&=-(n-2)! .\end{align}

\subsection{Continuum Hard Core Gas - Tonks Gas}
For a continuum hard core gas in one dimension with diameter $1$, the potential is:
\be \varphi(q_i,q_j)= \begin{cases} \infty &\text{ if } \; \; \; |q_i-q_j|<1 \\ 0 &\text{ otherwise } \end{cases} . \ee
The exponential and Mayer $f$-functions are:
\begin{align} \exp(-\beta \varphi(q_i,q_j))&= \begin{cases} 0 &\text{ if } \; \; \: |q_i-q_j|<1 \\ 1 &\text{ otherwise } \end{cases}. \\
 f(q_i,q_j)&= \begin{cases} -1 &\text{ if } \; \; \; |q_i-q_j|<1 \\ 0 &\text{ otherwise } \end{cases}. \end{align}
We therefore have the graph weight:
\be w(G)=(-1)^{e(G)}\int\limits_{\mathbb{R}^{n-1}} \! \prod\limits_{\{i,j\} \in E(G)}\chi(|x_i-x_j|<1)\,  \mathrm{d}x_2 \cdots \mathrm{d}x_n , \ee
where $x_1=0$ and $\chi$ is the indicator function. 

In \cite{DLL07}, this is interpreted as a the volume of a convex polytope $\Pi_G$ in $\mathbb{R}^{n-1}$. The polytope is defined by:
\be \Pi_G = \{ (\mathbf{x})_{[2,n]} \in \mathbb{R}^{n-1} \vert |x_i-x_j|<1 \, \forall \{i,j\} \in E(G) \, x_1=0 \}. \notag \ee
We use the notation $[2,n]=\{2, 3, \cdots, n\}$ and $(\mathbf{x})_{[2,n]}=(x_2, \cdots x_n)$.

Hence the graph weight may be written as:
\be w(G)=(-1)^{e(G)}\vol(\Pi_G) . \ee
The derivation of the cluster and virial expansions, using statistical mechanics, are more difficult in this case, but they are done in \cite{DLL07} and we achieve:
\begin{align} \beta P &=W(z)=\sum\limits_{n=1}^{\infty}\frac{(-n)^{n-1}z^n}{n!} \label{eq:TGclusterexp} \\
\beta P &= \frac{\rho}{1-\rho}=\sum\limits_{n=1}^{\infty} \rho^n \label{eq:TGvirialexp} , \end{align}
where $W(z)$ is the Lambert $W$-function. 

If we compare these to the results of Mayer's First and Second Theorems, \eqref{eq:pressureconnectedinv} and \eqref{eq:pressure2connectedinv}, we obtain the combinatorial relationships:
\begin{align} \sum\limits_{G \in \mathcal{C}[n]}(-1)^{e(G)}\vol(\Pi_G) &=(-1)^{n-1}n^{n-1} \\
\sum\limits_{G \in \mathcal{B}[n]}(-1)^{e(G)}\vol(\Pi_G) &=-n(n-2)! .\end{align}

\section{Results}
\label{sec:results}
The results of this article are the combinatorial interpretations of the cancellations in the alternating sums of weighted two-connected graphs.
\begin{theorem}[Combinatorial Identity from the one-particle hard-core model]
\label{thm:1PHC}
The difference of two-connected graphs with an even number of edges and an odd number of edges is given by the following formula: 
\be \sum\limits_{G \in \mathcal{B}[n]}(-1)^{e(G)}=-(n-2)!. \ee
\end{theorem}
This is proved through an involution $\Psi$, given in Section \ref{sec:1PHCinv}, which effectively pairs graphs differing by only one edge, leaving some small collection of graphs fixed, which give the $(n-2)!$ factor. 

The fixed graphs are formed from an increasing tree on the vertex set $[n-1]$ with the vertex $n$ adjacent to every other vetex. The number of increasing trees on $[n-1]$ is $(n-2)!$. The tree has $n-2$ edges and we add $n-1$ edges from the vertex labelled $n$ to achieve $2n-3$ edges. This gives the definite minus sign and the combinatorial factor. 

\begin{definition}
An increasing tree is a labelled tree on which the sequence of vertex labels along all paths from the vertex labelled $1$ to the leaves form increasing sequences. An example of such a graph is shown in Figure \ref{fig:increasingtree}.
\end{definition}

\begin{figure}[H]
\includegraphics[scale=0.2]{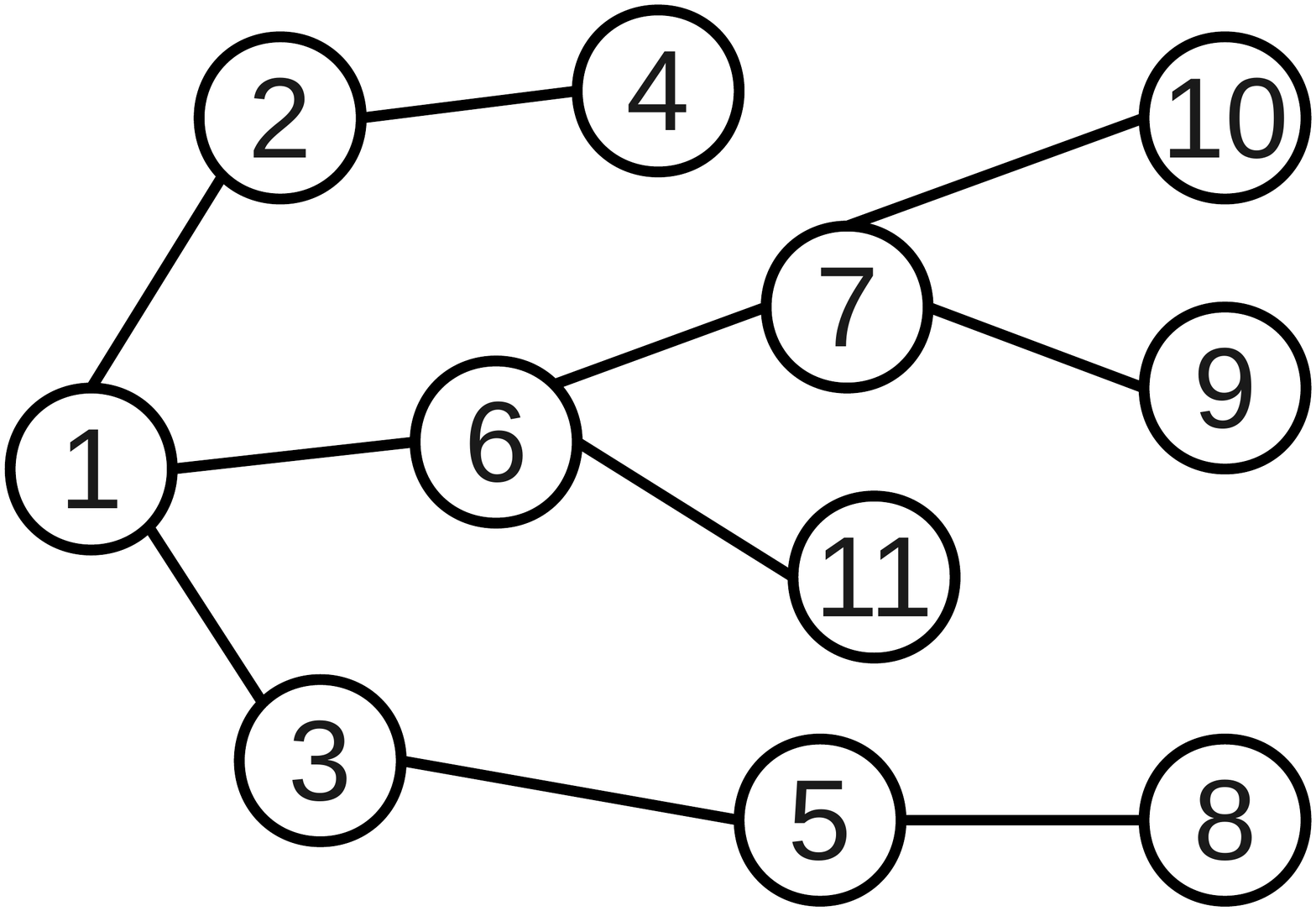} %eps1
\caption{An Increasing Tree on 11 vertices}
\label{fig:increasingtree}
\end{figure}
 
\begin{theorem}[Combinatorial Identity from the continuum hardcore gas]
\label{thm:CHCG}
When we add the polytope weights to the alternating graph sum, we achieve the following identity:
\be \sum\limits_{G \in \mathcal{B}[n]}(-1)^{e(G)}\text{Vol}(\Pi_G)=-n(n-2)! .\ee
\end{theorem}
This is proved through a collection of involutions $(\Psi_{\mathbf{h}})_{\mathbf{h} \in \mathbb{Z}^{n-1}}$. The index $\mathbf{h}$ is related to the partition of the polytopes into areas of equal volume attributed to Lass in \cite{B08, DLL07}. The meaning of $\mathbf{h}$ is explained in subsection \ref{subsec:lasspolytope}. The fixed points of these involutions occur only when $\mathbf{h}$ is of the form $(0,\cdots,0,-1, \cdots, -1)$, meaning that any edge is possible. There are precisely $n$ possibilities of these sequences, which corresponds to the $n$ positions of the last zero. 

The particular $\mathbf{h}$ provides a bijection $\sigma:[n] \to [n]$ on which the fixed graphs correspond to an increasing tree (given by the order $\sigma(i)<\sigma(j)$ if and only if  $i<j$) on the labels $\{\sigma(1), \cdots, \sigma(n-1)\}$. This is paired with every edge from $\sigma(n)$ to the vertices $\{\sigma(1), \cdots, \sigma(n-1)\}$. 

The number of these increasing trees on $n-1$ vertices is $(n-2)!$ and hence we obtain the factor $n(n-2)!$. We notice that these graphs are on $2n-3$ edges as above, which provides the minus sign. 

\begin{remark}[Complications for two-connected graphs]
The two-connected case is necessarily more complicated than the connected case. First of all, minimal two-connected graphs do not all have the same number of edges for a fixed number of vertices as trees (minimal connected graphs) do. Simply removing edges appropriately down to a minimal graph cannot provide a combinatorial understanding as there will still be sign differences to take care of. Furthermore, the sign of the factor is constant - the number of edges must always be odd for whatever value of $n$ we take. 
\end{remark}

\section{The Hardcore One Particle Gas - Proof of Theorem \ref{thm:1PHC}}
\label{sec:1PHCinv}
As indicated in the introduction, the proof of Theorem \ref{thm:1PHC} is done through an involution. To explain how the involution $\Psi$ provides the combinatorial factor through the number of fixed points, we use the manipulations of Bernardi \cite{B08}, where we know that the involution either adds an edge, removes an edge or leaves the graph fixed. We have that:
\be \sum\limits_{g \in \mathcal{B}[n]}(-1)^{e(g)}=\sum\limits_{g \in \mathcal{B}[n]}(-1)^{e(\Psi(g))}, \ee
since $\Psi$ is a bijection. The sum of these is therefore:
\begin{align} 2\sum\limits_{g \in \mathcal{B}[n]}(-1)^{e(g)}&=\sum\limits_{g \in \mathcal{B}[n]}((-1)^{e(g)}+(-1)^{e(\Psi(g))}) \notag \\ 
&= 2\sum\limits_{g \in \mathcal{B}[n] \vert \Psi(g)=g}(-1)^{e(g)} . \end{align}
The fixed points of the involution thus give us the combinatorial factor. 

This section describes the involution and proves it does what is required.

For graphs, the analogous operation to symmetric difference explained in the introduction is the operation $\oplus$. $G \oplus e$ is the graph $(V(G),E(G) \Delta \{ e\})$. 

The specific task of the proof of both identities is to identify for each graph a unique edge that we can add or remove. This has to be done in a consistent and efficient manner. Consistent in the sense that if we identify $e_G$ as the unique edge in $G$, then we want $e_{G \oplus e_G}=e_G$ so that $\Psi$ is an involution. It needs to be efficient in the sense that the only graphs it leaves fixed are those that provide the combinatorial factor relevant for the alternating sum. We do not want further cancellations to consider.

In each graph $G$, we consider the vertex labelled $n$.  When the vertex $n$ is adjacent to every vertex, we realise that the collection of two-connected graphs with this property may be identified with the collection of connected graphs on the vertex set $[n-1]$. Bernardi \cite{B08} has already provided an involution on this set that we can use in this case to obtain cancellations, since they will all come with the same prefactor $(-1)^{n-1}$ from the $n-1$ edges from the vertex labelled $n$. We thus firstly introduce the involution of Bernardi and make rigorous the connection between connected graphs on $[n-1]$ and the particular subset of two-connected graphs where $n$ is adjacent to every other vertex. 

For those graphs where the vertex labelled $n$ is not adjacent to every other vetex, we may use the two-connected property of the graph to find an edge suitable for the involution. This is done through using a corollary due to Whitney of Menger's theorem and introducing a definition of permissible edges. We emphasise how these combine to give a complete involution and that the only contributions arise from the Bernardi involution.

Firstly, we define the neighbourhood of a vertex $i$ in a graph $G$. 

\begin{definition}[Neighbourhood]
For a graph $G$ and a vertex $i$, we define the neighbourhood of $i$ in $G$ as $N_G(i):=\{j \in V(G) \vert \, \{i,j\} \in E(G) \}$.
\end{definition}

We define the \emph{lexicographic order} on edges $e \in [n]^{(2)}$ by:
\be \{i,j\} < \{k,l\} \text{ if } \begin{cases} \; \; \; \; \; \min\{i,j\} < \min\{k,l\} \\
\text{or }  \min\{i,j\} = \min\{k,l\} \text{ and } \max\{i,j\} < \max\{k,l\} \end{cases} . \notag \ee

For a subset $S$ of a totally ordered set, we define $S^{>e}:=\{x \in S \vert \, x>e\}$. For a graph $G=(V(G),E(G))$ and an edge $e$, we define $G^{>e}:=(V(G),E(G)^{>e})$ with respect to the lexicographic order above.

We give here Bernardi's involution on connected graphs, since it used for the two-connected graph version. We write it for the vertex set $[n-1]$ as this is the form in which it will be used.

\begin{definition}[Externally Active Edge]
An edge $e \in [n-1]^{(2)}$ is externally active for the graph $G \in \mathcal{C}[n-1]$, if there is a path in $G^{>e}$ between the endpoints of $e$.
\end{definition}

If a connected graph $G$ has an externally active edge, we define $\eta_G$ to be the maximal such edge.

\begin{definition}[Bernardi's Involution \cite{B08}]
The involution $\Psi_B: \mathcal{C}[n-1] \to \mathcal{C}[n-1]$ defined by Bernardi \cite{B08}, is given by:
\be \Psi_B:G \mapsto \begin{cases} G \oplus \eta_G &\text{ if } G \text{ has an externally active edge} \\ G &\text{ otherwise} \end{cases} .\ee
\end{definition}

The result of the involution is the following lemma.

\begin{lemma}[Bernardi \cite{B08}]
Under the involution $\Psi_B$, only increasing trees are kept fixed.
\end{lemma}

We introduce the following notation to simplify the formulation of the connection between two-connected graphs with vertex set $[n]$, where the vertex labelled $n$ is adjacent to every other vertex, and connected graphs with vertex set $[n-1]$.

\begin{itemize}
\item[i)] For a graph $G=(V(G),E(G))$, we denote by $G\setminus \{i\}$, \newline the graph $(V(G)\setminus \{i\},E(G) \setminus \{\{i,j\} \vert \, j \in V(G)\})$.
\item[ii)] We denote the subset of two-connected graphs on vertex set $[n]$ with vertex $n$ adjacent to all other vertices by $\mathcal{B}^{\Delta}[n]$.
\end{itemize}

\begin{definition}
We define the mapping $\zeta: \mathcal{B}^{\Delta}[n] \to \mathcal{C}[n-1]$, by $\zeta: G \mapsto G \setminus \{n\}$. We emphasise that removing a vertex and its incident edges from a two-connected graph leaves a connected graph and so defining the codomain of $\zeta$ as $\mathcal{C}[n-1]$ is fine.
\end{definition}
\begin{lemma}
The map $\zeta: \mathcal{B}^{\Delta}[n] \to \mathcal{C}[n-1]$ is a bijection.
\end{lemma}
\begin{proof}
Firstly it is injective. If $\zeta(G)=\zeta(H)$, this means $E(G)\cap [n-1]^{(2)}=E(H) \cap [n-1]^{(2)}$ and since $G$ and $H \in \mathcal{B}^{\Delta}[n]$, the remaining elements of $E(G)$ and $E(H)$, namely $\{\{i,n\} \vert \, i \in [n-1]\}$, are the same and so $G=H$. This is surjective, since for any connected graph on $[n-1]$, if we add the vertex labelled $n$ and all edges $\{i.n\}$ such that $i \in [n-1]$, the resulting graph is two-connected. If we consider removing any vertex $i \neq n$ from this new graph we see that every vertex is connected to every other vertex via $n$. If $n$ is removed then it is connected by definition and hence it is two-connected.
\end{proof}

We define the inverse map of $\zeta$ to be $\mu$.

\begin{definition}[Internally Disjoint Paths]
A path is an alternating sequence of vertices and edges in a graph $v_0 e_0 v_1 \cdots e_{k-1} v_k$, which begins and ends with a vertex. The edges are written in terms of the preceding and following vertices: $e_i=\{v_i,v_{i+1}\}$. Two paths $v_0 e_0 v_1 \cdots e_{k-1} v_k$ and $\tilde{v}_0 \tilde{e}_0 \tilde{v}_1 \cdots \tilde{e}_{l-1} \tilde{v}_l$ are internally disjoint if the only common vertices or edges are the endpoints $v_0=\tilde{v}_0$ and $v_k = \tilde{v}_l$.
\end{definition}

For an edge $\{i,j\}$ the endpoints are defined as the vertices $i$ and $j$. 

The following result in the case $k=2$ is used to find an edge in each graph where $N_G(n) \neq [n-1]$, by using the fact that we have two internally disjoint paths between $n$ and some $a \in [n-1] \setminus N_G(n)$. This is a classical theorem of Whitney \cite{W32} based upon Menger's Theorem. 

\begin{theorem}[Whitney \cite{W32}]
\label{thm:whitney}
A graph $G$ is $k$-connected if and only if every pair of vertices is connected by $k$ internally disjoint paths.
\end{theorem}

We introduce the notion of permissible edges as those edges which have both endpoints in the neighbourhood of a vertex $i$ and can easily be understood as a chord in the graph $G$, when we neglect any edges in $(N_G(i))^{(2)}$. We will focus on the case when $i=n$.

\begin{definition}[Permissible edges]
Given a (two-connected) graph $G$ and a vertex $i \in V(G)$, such that $N_G(i)\cup \{i\} \neq V(G)$, we define an edge $e \in [n]^{(2)}$ to be $(G,i)$-permissible if the following condition holds:

$\bullet$ there exists an $a \in V(G)\setminus (N_G(i) \cup \{i\})$, such that we have two vertex disjoint paths $a \to i$, intersecting each once in $N_G(i)$. The intersection vertices are the endpoints of $e$.
\end{definition}

If a two-connected graph $G$ with $V(G)=[n]$ has a $(G,n)$-permissible edge, then we denote the largest such edge in lexicographical ordering by $\varepsilon_G$.

\begin{lemma}
For every $G \in \mathcal{B}[n] \setminus \mathcal{B}^{\Delta}[n]$, we have a $(G,n)$-permissible edge.
\end{lemma}
\begin{proof}
We know that $S:=[n-1] \setminus N_G(n) \neq \emptyset$, because we are outside of $\mathcal{B}^{\Delta}[n]$. If we choose some $a \in S$, then we know by Theorem \ref{thm:whitney} we have two internally disjoint paths between the vertices labelled $a$ and $n$. Both paths must hit $N_G(n)$ at some point. When they first hit $N_G(n)$, then they could go straight to $n$ and so each path need only intersect $N_G(n)$ in one place. This provides us with a permissible edge and so $\varepsilon_G$ is well defined for every $G \in \mathcal{B}[n] \setminus \mathcal{B}^{\Delta}[n]$.
\end{proof}

\begin{definition}[Involution $\Psi$]
We define the involution $\Psi: \mathcal{B}[n] \to \mathcal{B}[n]$ through Bernardi's involution $\Psi_B$ and the permissible edge concept.
\begin{itemize}
\item[i)] If $N_G(n)=[n-1]$, we consider the graph $G \setminus \{n\}$. This is a connected graph and we may apply Bernardi's involution to this subgraph and retain the vertex $n$ and its incident edges.

This can be written as $\Psi\large\vert_{\mathcal{B}^{\Delta}[n]}:=\mu \circ \Psi_B \circ \zeta$.

\item[ii)] If $N_G(n) \neq [n-1]$, then we define the involution $\Psi: G \mapsto G \oplus \varepsilon_G$.
\end{itemize}
\end{definition}

The first point to emphasise is that due to the bijection between $\mathcal{B}^{\Delta}[n]$ and $\mathcal{C}[n-1]$, we are able to obtain cancellations for these graphs in the same way as Bernardi. We are left with increasing trees on the set $[n-1]$ and the vertex $n$ adjacent to every other vertex. 

We still need to prove that $\Psi$ is indeed an involution.

\begin{lemma}
$\Psi$ is an involution and its image is contained within $\mathcal{B}[n]$.
\end{lemma}

\begin{proof}
The fact this is true for $\Psi \large\vert_{\mathcal{B}^{\Delta}[n]}$ follows from the proof of Bernardi.

If an edge is permissible, we note that it is a chord in a cycle within the graph $G$. If we add an edge to a two-connected graph it remains two-connected. 

We prove below that if we remove a chord from a two-connected graph, then it remains two-connected.

We denote the chord we are considering by $c=\{i,j\}$,  the original graph by $G$ and the graph  $(V(G),E(G)\setminus \{c\})$ by $H$. We prove $H$ is two-connected by considering the effect of removing a vertex from $H$. There are two cases:
\begin{itemize}
\item[i)] $H \setminus \{i\}$ and $H \setminus \{j\}$ are connected as they are the same graphs as $G \setminus \{i\}$ and $G \setminus \{j\}$ respectively, which are connected since $G$ is two-connected.

\item[ii)]  If we consider another vertex $k$. We assume for contradiction that the graph $H \setminus \{k\}$ is not connected. We know $G \setminus \{k\}$ is connected and the only difference is that we have the additional edge $c$. This would then imply that $i$ and $j$ are in different connected components in $H \setminus \{k\}$. We know that $i$ and $j$ appear in a cycle in $H$. This means if we remove one vertex then we still have a path between $i$ and $j$, hence we obtain a contradiction unless $H \setminus \{k\}$ is connected.
\end{itemize}

The collection of permissible edges depends only on edges within  \newline $S:=[n-1] \setminus N_G(n)$, between $S$ and $N_G(n)$ and edges involving $n$. Adding or removing a permissible edge does not change the available edges on which one can make the two internally disjoint paths. Hence the collection of permissible edges for $G$ and $G \oplus \varepsilon_G$ are the same. This means that the largest elements in each set are the same i.e. $\varepsilon_{G \oplus \varepsilon_G}=\varepsilon_G$. Therefore it is an involution.
\end{proof}

Hence, $\Psi$ is an involution and has only fixed points in the set $\mathcal{B}^{\Delta}[n]$. The fixed points are those given by Bernardi as increasing trees on the vertex set $[n-1]$ with $n$ adjacent to all vertices in $[n-1]$.

\section{The Tonks Gas - Proof of Theorem \ref{thm:CHCG}}
\label{sec:CHCGinv}
In order to deal with the polytope volume weights, we decompose the polytopes into simplices. This first appeared in \cite{DLL07} and is used in \cite{B08} to prove the connected graph case. This splitting of polytopes into unimodular simplices is attributed to Lass.
\subsection{Polytopes and Simplices}
\label{subsec:lasspolytope}
This subsection explains how this splitting of the polytopes into simplices is used to construct the involution for the continuum case. These ideas are important in reducing the case of the Tonks gas to the one particle hard core model. 

The key idea is to split $\mathbb{R}^{n-1}$ into $(n-1)$-simplices of equal volume. We then realise that a polytope either fully contains a simplex, intersects only on the boundary of the simplex or is disjoint from the simplex. The sum is then reorganised so that we may sum over each simplex on the outside and then undertake the alternating sum on the restricted set of graphs whose associated polytopes contain the simplex considered. 

Consider $(\mathbf{x})_{[2,n]} \in \mathbb{R}^{n-1}$ and let $h_i$ be the integer part of $x_i$ and $0 \leq w_i<1$ be the fractional part such that $h_i+w_i=x_i$. Let $\sigma:[2,n] \to [2,n]$ be a bijection. We may define the simplex $\pi(\mathbf{h},\sigma)$, by the set of $\mathbf{x}$ with integer part $\mathbf{h}$ and whose fractional parts satisfy: $w_{\sigma(2)} <w_{\sigma(3)} < \cdots < w_{\sigma(n)}$. This simplex has volume $\frac{1}{(n-1)!}$. 

The condition $|x_i-x_j|<1$ is equivalent to $h_i-h_j \in \{0,\sign(w_j-w_i)\}$. We therefore have that $\pi(\mathbf{h},\sigma) \subset \Pi_G$ if and only if for all $\{i,j\} \in E(G)$, we have that $h_i-h_j \in \{0, \sign(\sigma^{-1}(j)-\sigma^{-1}(i))\}$ with $h_1=0$ and $\sigma(1)=1$.

\begin{lemma}
For any graph $G \in \mathcal{G}[n]$, the value $(n-1)!\vol(\Pi_G)$ counts the pairs $\mathbf{h} \in \mathbb{Z}^{n-1}$ and $\sigma \in S_{n-1}$ such that $\pi(\mathbf{h}, \sigma)$ is a subpolytope of $\Pi_G$. \end{lemma}

We may rearrange the sums over connected or two-connected graphs of the graph weights by first casting the sum as a sum over the pairs $(\mathbf{h},\sigma)$ and symmetrising the weight over isomorphic graphs. The symmetrisation procedure can be understood by considering a permutation $\sigma$ of $[2,n]$ and defining for any vector $\mathbf{h}=(h_2, \cdots, h_n) \in \mathbb{Z}^{n-1}$, $\sigma(\mathbf{h})=(h_{\sigma(2)}, \cdots, h_{\sigma(n)})$. For any graph $G$ with labels in $[n]$, the graph $\sigma(G)$ is the graph, with the same vertex set and satisfies $\{\sigma(i),\sigma(j)\} \in E(\sigma(G)) \iff \{i,j\}\in E(G)$. 

\begin{lemma}[Symmetrisation]
$\pi(\mathbf{h},\sigma) \subset \Pi_G$ if and only if $\pi(\sigma^{-1}(\mathbf{h}), \text{Id}) \subset \Pi_{\sigma(G)}$ for any permutation $\sigma$ of $[2,n]$. 
\end{lemma}
\begin{proof}
This equivalence can be elucidated by rewriting $\mathbf{w}=\sigma^{-1}(\mathbf{h})$ and $H=\sigma(G)$. This allows us to rewrite the latter statement as: $\pi(\mathbf{w},\mathrm{Id}) \subseteq \Pi_H$. This implies, for the entries in vector $\mathbf{w}$, that $\forall \{k,l\} \in E(H)$, $w_k-w_l \in \{0, \sign(l-k)\}$. Since $\{i,j\} \in E(G)\iff \{\sigma(i),\sigma(j)\}\in E(H)$, we may rewrite this as: $\forall \{i,j\} \in E(g)$ $w_{\sigma(i)}-w_{\sigma(j)} \in \{0, \sign(\sigma(j)-\sigma(i))\}$. We make the identification that $h_i=w_{\sigma(i)}$ to see that we get precisely the statement that $\pi(\mathbf{h}, \sigma) \subseteq \Pi_G$.
\end{proof} 

We let $\mathcal{H}$ denote either $\mathcal{C}$ or $\mathcal{B}$ and then we rewrite:
\begin{align} \sum\limits_{\substack{\mathbf{h} \in \mathbb{Z}^{n-1} \, G \in \mathcal{H}[n] \\ \pi(\mathbf{h}, \sigma) \subset \Pi_G}}(-1)^{e(G)} &=\sum\limits_{\substack{\mathbf{h} \in \mathbb{Z}^{n-1} \, G \in \mathcal{H}[n] \\ \pi(\sigma^{-1}(\mathbf{h}),\text{Id}) \subset \Pi_{\sigma(G)}}}(-1)^{e(G)} \notag \\
&= \sum\limits_{\substack{ \mathbf{h} \in \mathbb{Z}^{n-1} \, G \in \mathcal{H}[n] \\ \pi(\mathbf{h}, \text{Id}) \subset \Pi_G}}(-1)^{e(\sigma^{-1}(G))} \notag \\
&=\sum\limits_{\substack{\mathbf{h} \in \mathbb{Z}^{n-1} \, G \in \mathcal{H}[n] \\ \pi(\mathbf{h}, \text{Id}) \subset \Pi_G}}(-1)^{e(G)} \end{align}
We may therefore, understand the weight as:
\begin{align} \sum\limits_{G \in \mathcal{H}[n]}w(G)&=\sum\limits_{G \in \mathcal{H}[n]}(-1)^{e(G)}\vol(\Pi_G) = \frac{1}{(n-1)!}\sum\limits_{\substack{ \mathbf{h} \in \mathbb{Z}^{n-1} \sigma \in S_{n-1} \\ \text{such that }\pi(\mathbf{h},\sigma) \subset \Pi_G}}(-1)^{e(G)} \notag \\
&=\sum\limits_{\substack{\mathbf{h} \in \mathbb{Z}^{n-1} \, G \in \mathcal{H}[n] \\ \pi(\mathbf{h},\text{Id}) \subset \Pi_G}}(-1)^{e(G)} \end{align}
We define the \emph{centroid} of the vector $\mathbf{h}$, by $\bar{\mathbf{h}}=(\bar{h_1}, \cdots, \bar{h_n})$, where $\bar{h_i}=h_i+\frac{i-1}{n}$ and $\bar{h}_1=0$. We define $K_{\mathbf{h}}$ as the graph on $[n]$ where the edges are all pairs $\{i,j\}$ such that $|\bar{h_i}-\bar{h_j}| < 1$. We define $\mathcal{H}_{\mathbf{h}}[n]:=\{G \in \mathcal{H}[n] \vert E(G) \cap E(K_{\mathbf{h}})=E(G) \}$ where $\mathcal{H}$ can be replaced by $\mathcal{C}$ or $ \mathcal{B}$.

The final sum indicates that we need to count pairs $\mathbf{h}$ and $G$ such that $\pi(\mathbf{h}, \text{Id}) \subset \Pi_G$. That is that the centroid $\bar{\mathbf{h}} \in \Pi_G$, since $\bar{\mathbf{h}}$ is in the interior of $\pi(\mathbf{h},\text{Id})$. This can be recast as: for $\bar{\mathbf{h}} \in \Pi_G$, we require that:
\be \forall \{i,j\} \in E(G) \; |\bar{h}_i-\bar{h}_j|<1 \ee
We can, therefore, rewrite our sum as:
\be \sum\limits_{\mathbf{h} \in \mathbb{Z}^{n-1}}\sum\limits_{G \in \mathcal{H}_{\mathbf{h}}[n]}(-1)^{e(G)} \ee
we can thus consider the total sum as first a sum over the subset of graphs $\mathcal{H}_{\mathbf{h}}[n]$ for each $\mathbf{h}$ and add the results. This leads to considering separate $\Psi_{\mathbf{h}}:\mathcal{B}_{\mathbf{h}}[n] \to \mathcal{B}_{\mathbf{h}}[n]$ which are involutions and finding their fixed points.

\subsection{The Involutions $\Psi_{\mathbf{h}}$}

We define an involution $\Psi_{\mathbf{h}}$ for each $\mathbf{h} \in \mathbb{Z}^{n-1}$ on the set $\mathcal{B}_{\mathbf{h}}[n]$ of two connected graphs, which are compatible with the vector $\mathbf{h}$. We note that, by the definition of $\mathcal{B}_{\mathbf{h}}[n]$, edges with $|\bar{h}_i-\bar{h}_j|>1$ are forbidden. We call an edge $\{i,j \}$ such that $|\bar{h}_i-\bar{h}_j|<1$ allowed.

In order to make the connection with the proof in the discrete case, we indicate a bijection $\xi_{\mathbf{h}}$ related to the particular $\mathbf{h}$ that provides a suitable relabelling of the vertices to allow for an efficient application of the lemmas of section \ref{sec:CHCGinv} to a relabelled graph. We reframe the consequences of these lemmas in the context of the allowed edges. It is important to check that an edge we may want to add or remove by the prescription in section \ref{sec:CHCGinv} is allowed within the specific collection of graphs $\mathcal{B}_{\mathbf{h}}[n]$. It is then proved that when we have a non empty set of forbidden edges, all terms cancel. In the case when the set of forbidden edges is empty, we obtain the exact values taken by $\mathbf{h}$ and everything reduces to the discrete gas case with a relabelling. 

We have a definite order on the entries of $\bar{\mathbf{h}}$, since each entry has a different fractional part. We define a re-ordering of the set $[n]$, through a bijection $\xi_{\mathbf{h}}:[n] \to [n]$. This re-ordering is defined through the order for the entries of $\bar{\mathbf{h}}$: $\bar{h}_{\xi_{\mathbf{h}}(1)}< \bar{h}_{\xi_{\mathbf{h}}(2)}< \cdots < \bar{h}_{\xi_{\mathbf{h}}(n)}$. 

The re-ordered lexicographic order on edges is given by:
\be \{\xi_{\mathbf{h}}(i),\xi_{\mathbf{h}}(j)\} < \{\xi_{\mathbf{h}}(k),\xi_{\mathbf{h}}(l)\} \text{ if } \begin{cases} \; \; \; \; \; \min\{i,j\} < \min\{k,l\} \\
\text{or }  \min\{i,j\} = \min\{k,l\} \text{ and } \max\{i,j\} < \max\{k,l\} \end{cases} . \notag \ee

Instead of considering $(G,n)$-permissible edges, we consider $(G,\xi_{\mathbf{h}}(n))$-permissible edges since it makes the formulation of the involution easier. 
\begin{lemma}
\label{lem:permisallowed}
All edges $e \in N_G(\xi_{\mathbf{h}}(n))^{(2)}$ are allowed.
\end{lemma}
\begin{proof}
We realise that $\forall i \in N_G(\xi_{\mathbf{h}}(n))$ we have that $\bar{h}_{\xi_{\mathbf{h}}(n)}-1< \bar{h}_i< \bar{h}_{\xi_{\mathbf{h}}(n)}$ and so for every pair $i,j \in N_G(\xi_{\mathbf{h}}(n))$, $|\bar{h}_i-\bar{h}_j|<1$ and hence the edge is allowed in $\mathcal{B}_{\mathbf{h}}[n]$. 
\end{proof}
\begin{corollary}
All $(G,\xi_{\mathbf{h}}(n))$-permissible edges are allowed.
\end{corollary}
\begin{lemma}
If $N_G(\xi_{\mathbf{h}}(n))=\xi_{\mathbf{h}}([n-1])$, then $\mathcal{B}_{\mathbf{h}}[n]=\mathcal{B}[n]$ and $\mathbf{h}$ is of the form of an initial sequence of zeroes with remaining entries $-1$.
\end{lemma}
\begin{proof}
By lemma \ref{lem:permisallowed}, all edges in $\xi_{\mathbf{h}}([n-1])^{(2)}$ are allowed. The edges $\{ \xi_{\mathbf{h}}(n), j \}$ for all $j \in \xi_{\mathbf{h}}([n-1])$ are already in the graph and so cannot be forbidden. Hence every edge is allowed and so $\mathcal{B}_{\mathbf{h}}[n]=\mathcal{B}[n]$.

Since $\bar{h}_1=0$, this means $\bar{h}_i \in (-1,1)$ for all $i$. We also note that if $\bar{h}_j<0$, then $\bar{h}_k<0$ for all $k>j$. This arises from the fact that the entries of $\mathbf{h}$ are restricted to $\{-1,0\}$. For a negative entry we will have $\bar{h}_j=-1+\frac{j-1}{n}$, which is not within distance $1$ of the value $0+\frac{k-1}{n}$ for any $k>j$. This means that $\mathbf{h}$ is of the special form of an initial sequence of zeroes with the remaining entries $-1$. 
\end{proof}

\begin{definition}
We define the set $\mathcal{B}^{\Delta_{\mathbf{h}}}[n]$ as the collection of two-connected graphs where $\xi_{\mathbf{h}}(n)$ is adjacent to all other vertices. We have the corresponding maps $\zeta_{\mathbf{h}}$ and $\mu_{\mathbf{h}}$ between $\mathcal{B}^{\Delta_{\mathbf{h}}}[n]$ and $\mathcal{C}[[n] \setminus \{\xi_{\mathbf{h}}(n)\}]$, which are the same as in section \ref{sec:CHCGinv}, except we are removing the vertex $\xi_{\mathbf{h}}(n)$ instead of $n$.

Formally, we can write these bijections as a conjugation with $\xi_{\mathbf{h}}$, when interpreted as its action on graphs. In this case:
\begin{align} \zeta_{\mathbf{h}} &:= \xi_{\mathbf{h}} \circ \zeta \circ \xi_{\mathbf{h}}^{-1} \\
\mu_{\mathbf{h}} &:= \xi_{\mathbf{h}} \circ \mu \circ \xi_{\mathbf{h}}^{-1} \end{align}
\end{definition}

\begin{definition}
We define the modified Bernardi involution $\Psi_{B, \mathbf{h}}$ as in section \ref{sec:CHCGinv}, except $G^{>e}$ is interpreted in the sense of the re-ordered lexicographic ordering and for $\eta_G$ to be maximal externally active edge we use this ordering too. This can also be simply written using the graphical label conjugation:
\be \Psi_{B, \mathbf{h}} :=  \xi_{\mathbf{h}} \circ \Psi_B \circ \xi_{\mathbf{h}}^{-1} \ee
\end{definition}

The largest (using the re-ordered lexicographic order) $(G, \xi_{\mathbf{h}}(n))$-permissible edge is denoted by $\varepsilon_{G, \mathbf{h}}$.

\begin{definition}
We define $\Psi_{\mathbf{h}}$ as the involution on $\mathcal{B}_{\mathbf{h}}[n]$, defined by:
\begin{itemize}
\item[i)] If $N_G(\xi_{\mathbf{h}}(n))=\xi_{\mathbf{h}}([n-1])$, then we may use a modified version of Bernardi, since all edges are possible in $\mathcal{B}_{\mathbf{h}}$.  

\be \Psi_{\mathbf{h}} \large\vert_{\mathcal{B}^{\Delta_{\mathbf{h}}}[n]} := \mu_{\mathbf{h}} \circ \Psi_{B, \mathbf{h}} \circ \zeta_{\mathbf{h}} \notag \ee

\item[ii)] Otherwise, we have a permissible edge and can perform the involution $\Psi_{\mathbf{h}}: G \mapsto G \oplus \varepsilon_{G, \mathbf{h}}$.
\end{itemize}
\end{definition}
$\Psi_{\mathbf{h}}$ retains the property of being an involution on two-connected graphs as in section \ref{sec:CHCGinv}.

We are thus left with only those graphs that have $N_G(\xi_{\mathbf{h}}(n))=\xi_{\mathbf{h}}([n-1])$ and are increasing with respect to the re-ordered lexicographic order. The only $\mathbf{h}$ vectors that contribute are those with an initial sequence of zeros followed by $-1$s. There are $n$ possibilities of these sequences, since the final $0$ can appear in any of the entries  $h_1 \cdots h_n$. 
\begin{lemma}
The permutation $\xi_{\mathbf{h}}$ related to the $\mathbf{h}$-vector with $h_j=0$ for $1 \leq j \leq s$ and $h_k=-1$ for $k>s$ takes the special form:
\be \xi_{\mathbf{h}}: i \mapsto i- s \, \mathrm{mod}\,  n. \notag \ee
\end{lemma}
\begin{proof}
We observe that the entry $\bar{h}_{s+1}$ has the smallest value so $s+1 \mapsto 1$. We then note that the following entries are negative and are in increasing order. The preceding entries are also in increasing order but are positive. Hence we have $\xi_{\mathbf{h}}(s+k)=k$ for $0 \leq k \leq n-s$ and $\xi_{\mathbf{h}}(i)=n-s+i$ for $1 \leq i \leq s-1$, which can be written in the form in the lemma for brevity.
\end{proof}
 Hence we have a precise collection of two-connected graphs. We have the examples from section \ref{sec:CHCGinv} with these linear relabellings. 

\begin{lemma}
For $n \geq 5$, the fixed graphs are all distinct.
\end{lemma}

\begin{proof}
We indicate that there are no labelled graph automorphisms of the form of $\xi_{\mathbf{h}}$ described above for the increasing trees on $[n-1]$ with $n$ adjacent to every vertex. The first observation is that $\xi_{\mathbf{h}}$ has no fixed vertex labels. We know the degree of the vertex labelled $n$ is $n-1$. If we were to have an automorphism with no fixed labels, then we require another vertex of the same degree to send $n$ to. This means we need a vertex in the increasing tree adjacent to all other vertices in the increasing tree. 

When a tree has at least three vertices, only one vertex can be adjacent to the rest, since if we have two vertices adjacent to all vertices we have them adjacent to each other and some third vertex. This creates a $3$-cycle contradicting the fact a tree is acyclic. Furthermore, in this increasing tree, this vertex can only be the vertex labelled $1$ or $2$. For any $k \in [3,n-1]$, $k$ cannot be attached to both $1$ and $2$, or else we will have a $3$-cycle, as we always have the edge $\{1,2\}$. 

We therefore require that the graph automorphism exchanges the labels of the two vertices. The automorphisms are translations and since $n \mapsto 1$ or $n \mapsto 2$, we have to translate by $1$ or $2$, but then the vertex labelled $1$ or $2$ would not be relabelled as $n$ as we would require. Hence $\xi_{\mathbf{h}}$ is not an automorphism for any of the prescribed graphs and so the collection of these graphs for $n \geq 5$ are all distinct.
\end{proof}

\section{The Structure of Two-connected Graphs}
\label{sec:structure}
In this section, we indicate how the structure of two-connected graphs indicates the importance of graphs with $2n-3$ edges. Firstly, we explain some preliminary concepts about block cutpoint trees and then use these to explain why minimal two-connected graphs, that is a two-connected graph, such that the removal of an edge renders the graph no longer two-connected, on $n$ vertices have at most $2n-4$ edges. 
\subsection{The Block Cutpoint Tree}
\label{subsec:bctree}
In this section, we introduce the notion of a block cut-point tree and state a result relating the number of vertices in the individual blocks to the number of vertices in the whole graph. We use the notation $\mathfrak{a}$ to denote the collection of trees.
\begin{itemize}
\item An \emph{articulation point} in a connected graph is a vertex, which when it and its incident edges are removed, renders the graph disconnected. A synonym that is frequently used is a \emph{cutpoint}.
\item A \emph{two-connected graph} is a connected graph without articulation points.
\item A \emph{block} is a maximal two-connected subgraph of a connected graph. Maximal in terms of edges and vertices it includes. 
\end{itemize}
The block cutpoint tree (bc-tree) associated to a connected graph $G$ is a (bipartite) graph where the vertices represent the articulation points and the blocks in a connected graph. An edge, between an articulation point and a block, is present in this graph, when an articulation point is contained in a block. It is a tree, since if there were a cycle in this graph then the cycle itself would have been a block. An example of a block cutpoint tree is shown in Figure \ref{fig:bctreethesis}.

\begin{definition}[The Centre of a Tree]
To define the centre of a tree formally, we define first the {\it eccentricity} $\varepsilon(v)$ of a vertex $v$ as the minimal graph distance of $v$ to a leaf. This may be formally written as $\varepsilon(v):=\min \{d_H(v,l) \vert \, \mathrm{deg}(l)=1\}$, where $d_H$ indicates the Hamming or graph distance in the tree.

The centre of a tree is the collection of vertices at which the maximum eccentricity is attained. This can either be two neighbouring vertices or a single vertex. In the former case, we often call the edge between the vertices the centre of the tree.  
\end{definition}

\begin{remark}[An Algorithmic Interpretation of the Centre of the Tree]
One can apply the function $f:\mathfrak{a} \to \mathfrak{a}$, which for any given tree, removes all leaves and the edges incident to the leaves. Formally, we can write this as:
\be f:(V(\tau),E(\tau)) \mapsto (V(\tau) \setminus L, E(\tau) \setminus (L \times V(\tau))) ,\ee
where $L:=\{ i \in V(\tau) \vert \mathrm{deg}(i)=1 \}$, the collection of leaves. 

Repeated application of $f$, gives a sequence of trees, $(f^n(\tau))_{n \in \mathbb{N}_0}$ which becomes constant either when we have a single vertex or the empty graph. In the case of the single vertex, this is the centre of the tree. For the empty graph, the penultimate step will have been two vertices and an edge. This edge or the pair of vertices is defined as the centre.
\end{remark}

A bc-tree is bipartite with all leaves in one set (the blocks). It therefore has a unique centre, since the eccentricity of the articulation points will be odd and the eccentricity of the blocks will be even so two neighbours cannot have the same maximum eccentricity. Since we have a unique vertex at the centre of the bc-tree, we may define a digraph arising from the bc-tree, where the edge is oriented to point away from the centre. An example is displayed in Figure \ref{fig:digraph}.

\begin{figure}[here]
\begin{minipage}[b]{0.5\linewidth}
\centering
\includegraphics[scale=0.25]{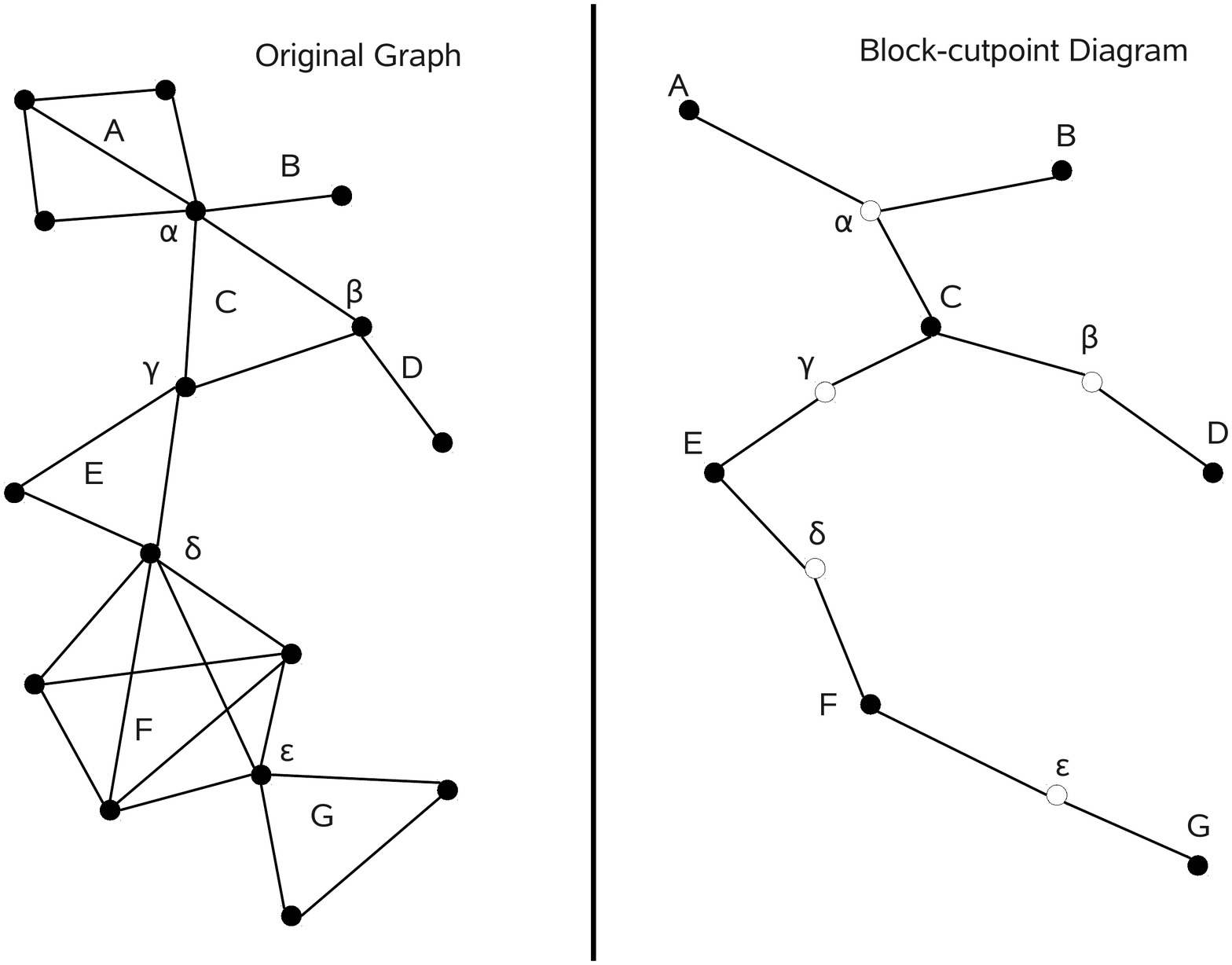} %esp
\caption{An example of a bc-tree}
\label{fig:bctreethesis}
\end{minipage}
\hspace{0.3cm}
\begin{minipage}[b]{0.45\linewidth}
\centering\includegraphics[scale=0.16]{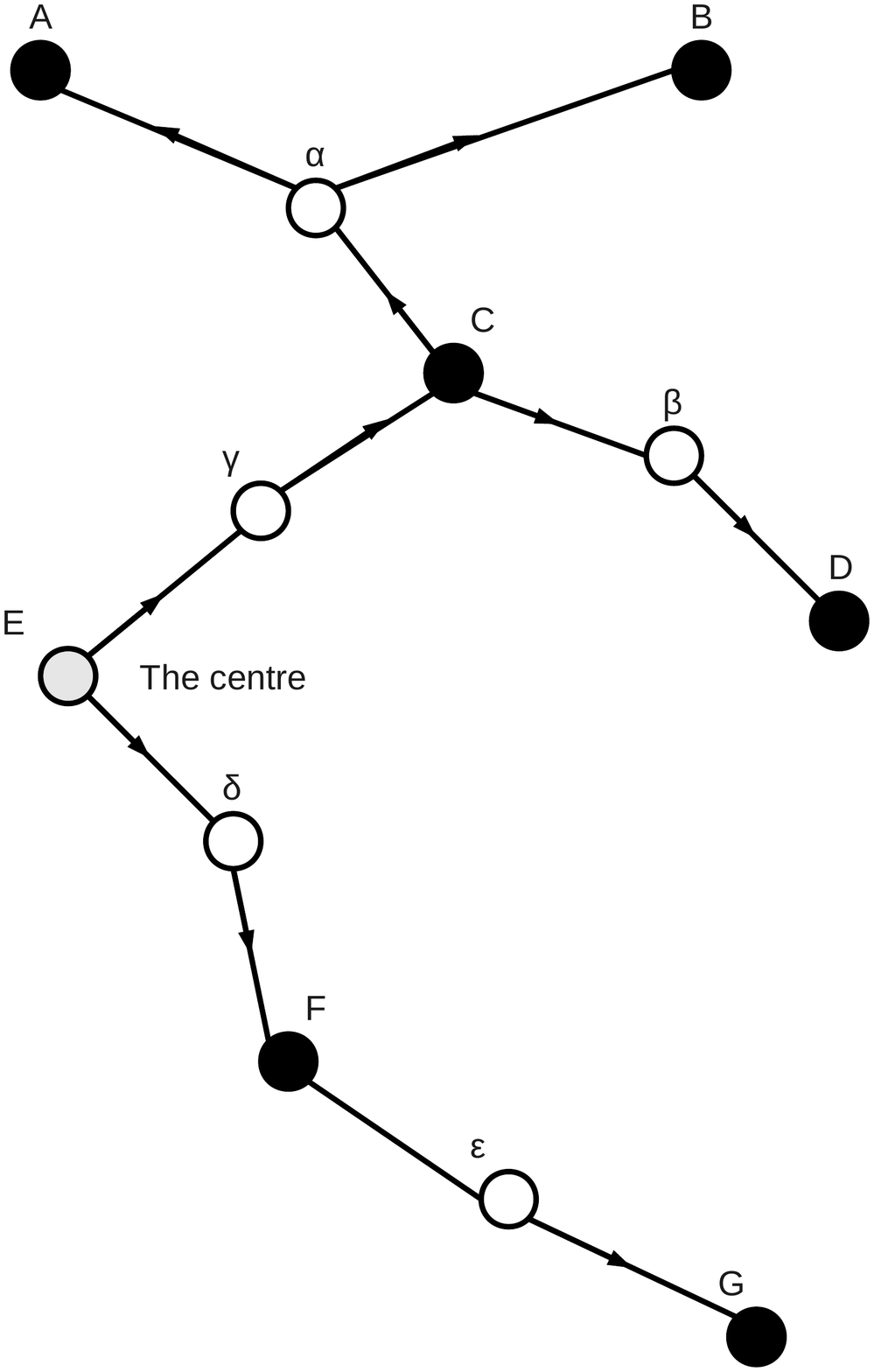}
\caption{The associated digraph}
\label{fig:digraph}
\end{minipage}
\end{figure}

\begin{lemma}[Block Decomposition]
\label{lem:blockdecomp}
If we decompose a connected graph on $n$ vertices into its block structure and let $I$ index the collection of blocks and $(k_i)_{i \in I}$ be the sequence of block sizes, then we have the following equality:
\be \sum\limits_{i \in I}(k_i-1)=n-1 \label{eq:bctreeidentity} . \ee
\end{lemma}
\begin{proof}
The key idea is to indicate what vertex we omit inside each block on the left hand side of \eqref{eq:bctreeidentity}. The digraph  gives an (essentially) unique prescription of the missing vertex in each block and in which block an articulation point is counted.

The digraph comprises of two types of directed edge $(B,a)$ and $(a,B)$, where $a$ indicates an articulation point and $B$ a block. The arrow points from the first entry to the second entry. Since there is a unique path from the centre to every other vertex, every vertex has precisely one edge in which they are the second entry. 

There are two key cases:
\begin{itemize}
\item[i)] \textbf{The centre is an articulation point}

For a block, $B$, the unique vertex we neglect on the left hand side of \eqref{eq:bctreeidentity} is the articulation point, $a$, where $(a,B)$ is the directed edge in the digraph. 

Every articulation point, $\alpha$, except the centre appears in an edge $(\beta, \alpha)$, for which it is the second entry, meaning it is enumerated in the left hand side of \eqref{eq:bctreeidentity} in precisely one block. The central articulation point is the only neglected vertex, which gives the right hand side of \eqref{eq:bctreeidentity}.

\item[ii)] \textbf{The centre is a block}

In this case every block, except the centre, can be given the prescription as for the first case. For the central block, we can choose precisely one of its neighbours to neglect. All articulation points in this case have an edge in which they are the second entry and so are counted, excepting the articulation point identified by the central block. Therefore, we have \eqref{eq:bctreeidentity}.
\end{itemize}
\end{proof}

\subsection{The Importance of $2n-3$ Edges}

To understand why the two-connected graphs on $n$ vertices with $2n-3$ edges play a special role, we first indicate that two-connected graphs with at least this number of edges cannot be minimal.

Given a graph $G$ on the vertex set $[n]$, we denote by $d_1$, the degree of the vertex labelled $1$.
 
\begin{lemma}
Two-connected graphs on $n$ vertices with $\geq 2n-3$ edges are not minimal, that is they necessarily have a chord.
\end{lemma}
\begin{proof} 
This is done by induction on the number of vertices $n$. 

The cases $n=2,3$ are vacuous and one can see from the examples in Figure \ref{fig:5edges4vertices} that this holds when $n=4$.

The connected graph $G\setminus \{1\}$ may be decomposed into its bc-tree. Each block with $l$ vertices in the tree has to have $\leq 2l-4$ edges or else we have a smaller graph which has a chord by induction. We note here that blocks of size $2$ or $3$ need to be treated separately. We let $l_i$ denote the size of the $i$th block not of size $2$ or $3$ and $b_2$ and $b_3$ denote the number of blocks of size $2$ and $3$ respectively. We have from lemma \ref{lem:blockdecomp}:
\be \sum_i(l_i-1)+b_2+2b_3=n-2 \ee
The total number of edges in $G \setminus \{1\}$ must then not exceed:
\be \sum_i2(l_i-1)-2b_{\geq 4}+b_2+3b_3 \leq 2n-4 - b_2-b_3-2b_{\geq 4} \ee
where $b_{\geq 4}$ indicates the number of blocks with more than four vertices.
We know that $e(G \setminus \{1\}) \geq 2n-3-d_1$ and so we obtain the inequality:
\be d_1 \geq 1+b_2+b_3+2b_{\geq 4} \geq 1+ \text{ total number of blocks} \ee

If we have only one block, then we either have two neighbours of $1$ and can apply induction to this block, as it will be a two-connected graph on $n-1$ vertices and at least $2(n-1)-3$ edges.

If we have at least three neighbours of $1$ in a block, say $\alpha$, $\beta$ and $\gamma$, then we may find a path $\alpha \to \beta \to \gamma$. This follows from Theorem \ref{thm:whitney}, since we have two internally disjoint paths between $\alpha$ and $\beta$ and between $\beta$ and $\gamma$. If we go along one of the paths between $\alpha$ and $\beta$ until we first hit one of the two paths between $\beta$ and $\gamma$, from here we follow the path towards $\beta$ and then take the disjoint path to $\gamma$, this is then a path between $\alpha$ and $\gamma$ that goes via $\beta$ and does not self-intersect. In this case $\{1,\beta\}$ forms a chord.

The final case is if we have at least two blocks and at most two neighbours of $1$ in a block. Then we have a block with two neighbours of $1$ call them $\alpha$ and $\beta$ and we have a third neighbour of $1$, $\gamma$ in some other block. Let $A$ be the articulation point of the block containing $\alpha$ and $\beta$ closest to $\gamma$. We have a path from $A$ to $\gamma$ outside of this block since it is a connected graph. We are also able to construct a path $\alpha \to \beta \to A$ since they are all in one block. Concatenating these paths gives again a path $\alpha \to \beta \to \gamma$ from which we determine $\{1,\beta\}$ is a chord. 

\end{proof}
It is also possible to construct a graph with $n$ vertices and $2n-4$ edges that is minimally two-connected, as shown in Figure \ref{fig:minimally2conn}. The number of edges being $2n-3$ marks some transition in the possibility of being minimal.

\begin{figure}[H]
\begin{minipage}[b]{0.5\linewidth}
\centering
\includegraphics[scale=0.2]{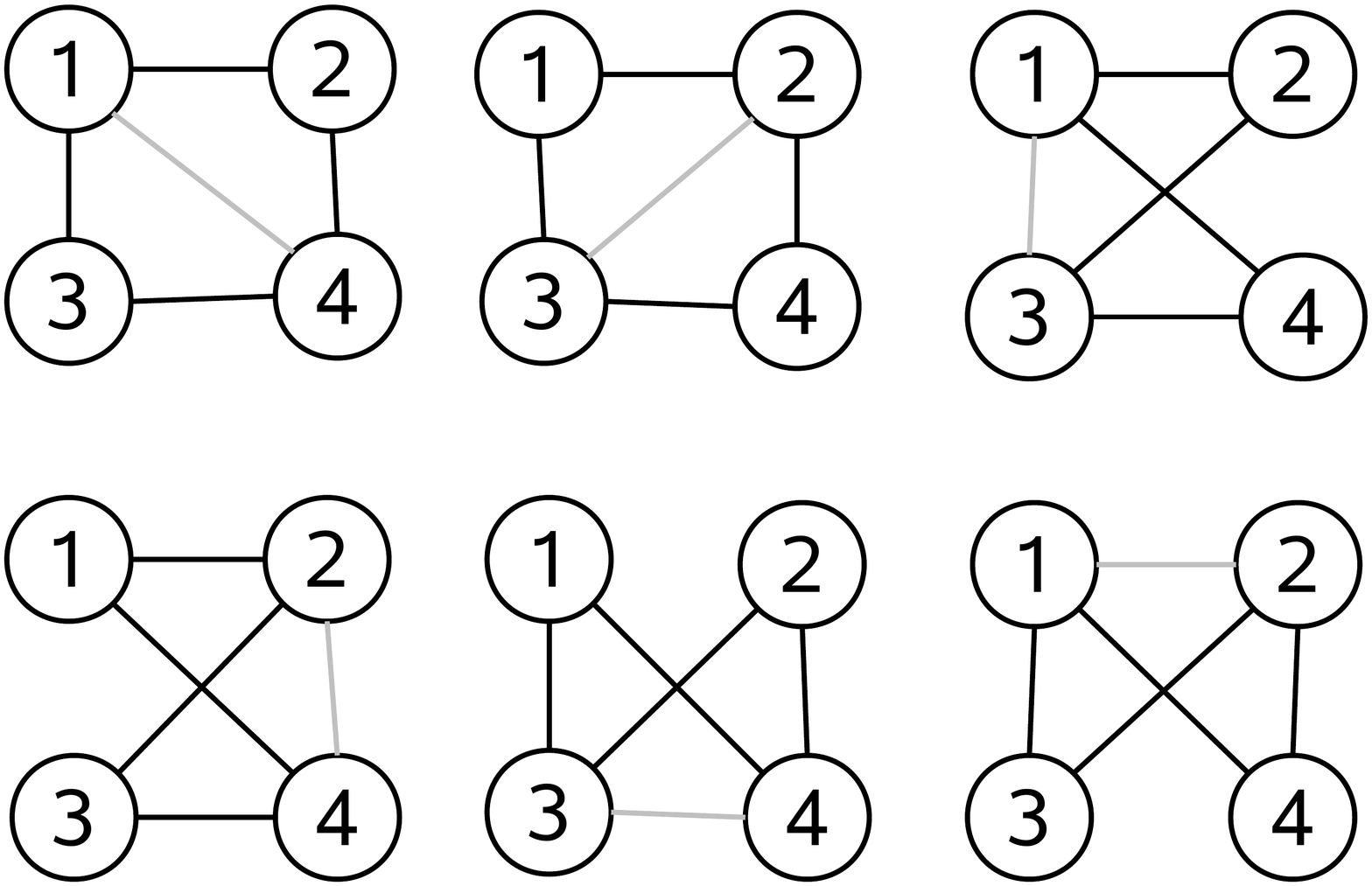} %eps1
\caption{The chords in graphs of $n=4$ vertices, excluding the complete graph }
\label{fig:5edges4vertices}
\end{minipage}
\hspace{0.3cm}
\begin{minipage}[b]{0.45\linewidth}
\centering
\includegraphics[scale=0.15]{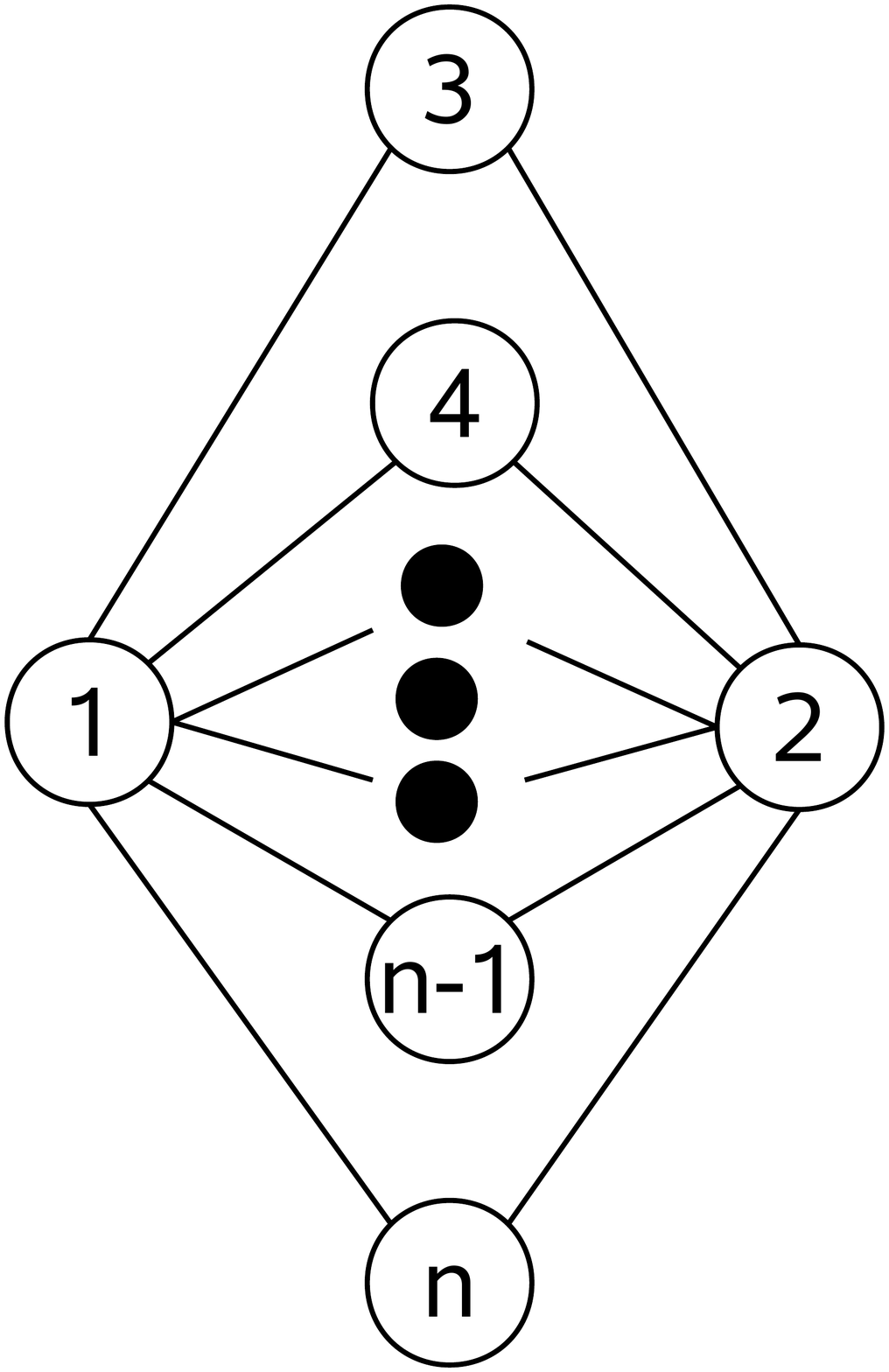} %eps1
\caption{A graph with $2n-4$ edges and $n$ vertices which is minimally two-connected }
\label{fig:minimally2conn}
\end{minipage}
\end{figure}

\section{The Connections with Tree Identities}
\label{sec:extend}
In this section, we convey the connection between involutions and partition schemes for connected graphs and how the latter is used to give estimations of the coefficients in the expansions. This is used as motivation to consider whether the two-connected graph involution may have such a connection.

The paper \cite{FP07} presents the notion of the partition in the sense of Penrose and gives the general idea of a partition. We define a partial order of $\mathcal{C}[n]$ by bond inclusion: $G \leq \tilde{G}$ $\iff$ $E(G) \subset E(\tilde{G})$. 
For $G \leq H$, we define the set $[G,H]=\{K\vert \, G \leq K \leq H\}$
The Penrose construction partitions the set of connected graphs into subsets of the form $[\tau, R(\tau)]$, where $R: \mathfrak{a}[n] \to \mathcal{C}[n]$. Many different constructions can be used to achieve an $R$. Penrose gave one explicit example in \cite{P67}.

\begin{definition}[Partition Scheme]
A partition scheme for the set of connected graphs $\mathcal{C}[n]$ is any map $R:\mathfrak{a}[n] \to \mathcal{C}[n]$ $\tau \mapsto R(\tau)$, such that:
\begin{itemize}
\item[i)] $E(R(\tau)) \supset E(\tau)$ and
\item[ii)] $\mathcal{C}[n]$ is the disjoint union of the sets $[\tau,R(\tau)]$ for $\tau \in \mathfrak{a}[n]$.
\end{itemize}
\end{definition}
\noindent The Penrose scheme is as follows:

For any vertex $i$ of $\tau \in \mathfrak{a}[n]$, we denote by $d(i)$ the tree distance between the vertices $i$ and $1$. We let $i'$ be the predecessor of $i$ i.e. $d(i')=d(i)-1$ and $\{i',i\} \in E(\tau)$. We associate to $\tau$, the graph $R_{\mathrm{Pen}}(\tau)$ found by adding (only once) to $\tau$ all edges $\{i,j\} \in [n]^{(2)}$ such that either:
\begin{itemize}
\item[P1] $d(i)=d(j)$ edges between vertices at same generation
\item[P2] $d(j)=d(i)-1$ and $i'<j$ edges between vertices one generation away
\end{itemize}
For a partition scheme $R$, denote by $\mathfrak{a}_R:=\{\tau \in \mathfrak{a}[n] \vert \, R(\tau)=\tau \}$ the set of $R$-trees. In particular, $\mathfrak{a}_{R_{\mathrm{Pen}}}$ is the set of Penrose trees. 

The following proposition emphasises where the Boolean partition offers advantages to providing estimations.

\begin{proposition}[Bounding the Connected Graph Sum]
In models where we have soft repulsion (a positive potential), the Mayer $f$-function satisfies $|1+f_e| \leq 1$. Using a partition scheme, we have the bound:
\be \left\vert \sum\limits_{G \in \mathcal{C}[n]}\prod\limits_{e \in E(G)}f_e \right\vert \leq \sum\limits_{\tau \in \mathfrak{a}[n]} \prod\limits_{e \in E(\tau)}|f_e| \leq |\mathfrak{a}[n]| \ee
\end{proposition}
\begin{proof}
For any numbers $(f_e)_{e \in [n]^{(2)}}$, we have :
\begin{align} \sum\limits_{G \in \mathcal{C}[n]} \prod\limits_{e \in E(G)} f_e &= \sum\limits_{\tau \in \mathfrak{a}[n]} \prod\limits_{e \in E(\tau)}f_e \sum\limits_{F \subset E(R(\tau)) \setminus E(\tau)} \prod\limits_{e \in F} f_e \notag \\
&=\sum\limits_{\tau \in \mathfrak{a}[n]}\prod\limits_{e \in E(\tau)}f_e \prod\limits_{e \in E(R(\tau))\setminus E(\tau)}(1+f_e) \label{eq:penroserewrite} \end{align}
When we take the absolute value of the right hand side, we may use the triangle inequality and bound the second product in \eqref{eq:penroserewrite} by $1$. 
\end{proof}

In the hardcore case, the second product in \eqref{eq:penroserewrite} is zero unless $R(\tau)=\tau$, giving that the fixed points of this $R$ function also give a combinatorial interpretation of the cancellations.The alternative combinatorial interpretation of fixed points provided by Penrose trees is that, considering the tree as being rooted at $1$, we are required to have precisely one vertex in each generation. This necessarily gives a linear tree. We have to determine the positions of $i \in [2,n]$, which are defined uniquely by their distance from $1$, which corresponds to a bijection, $f: [2,n] \to [n-1]$, giving the $(n-1)!$ factor. 

To define the Penrose involution arising from the Penrose construction, we make the following definition of a Penrose active edge.

For a graph $G$, we define the Hamming distance between vertices labelled $i$ and $j$ as $d_G(i,j)$ which is the length of the shortest path between $i$ and $j$.

\begin{definition}[Penrose Active Edges]
An edge $\{i,j\}$ is called \emph{Penrose active} for $G$ if, either:
\begin{itemize}
\item[i)] $d_G(1,i)=d_G(1,j)$  or
\item[ii)]  $d_G(1,i)=d_G(1,j)+1$ and $\exists i'<j$ such that $\{i,i'\} \in E(G)$ with $d_G(1,i')=d_G(1,j)$.
\end{itemize}
\end{definition}

We let $e_{\mathrm{Pen},G}^{\star}$ be the greatest Penrose active edge for $G$ in lexicographic order. 
\begin{lemma}[The Penrose Involution]
The mapping:
\be \Psi_{\mathrm{Pen}}: G \mapsto \begin{cases} G \oplus e_{\mathrm{Pen},G}^{\star} &\text{ if } G \text{ has a Penrose active edge} \\ G &\text{ otherwise} \end{cases} \ee
is an involution on connected graphs. 
\end{lemma}

\begin{proof}
We first prove that $d_G(1,k)=d_{\Psi_{\mathrm{Pen}}(G)}(1,k)$. The two graphs $G$ and $\Psi_{\mathrm{Pen}}(G)$ differ only on an edge $e^{\star}_{\mathrm{Pen}, G}=:\{i,j\}$,  where $|d_G(1,i)-d_G(1,j)|\leq 1$. Throughout this proof in the case where we have equality, we assume without loss of generality that $d_G(1,i)+1=d_G(1,j)$.

For any $k$, we consider the distance from the vertex labelled $k$ to $1$ in both graphs. This is defined through the shortest path from $1$ to $k$. We indicate that for any path between $1$ and $k$ containing the edge $\{i,j\}$ we can find a path of the same or shorter length that does not contain this edge. If $d_G(1,j)=d_G(1,i)$, then considering a path from $1$ to $k$ up to this edge, we realise that the shortest length the path up to this edge can be is $d_G(1,j)+1$, but we know that there is a shorter path to this endpoint because $d_G(1,j)=d_G(1,i)$ and so we can replace this initial path with a shorter path. 

We are left with the case $d_G(1,j)+1=d_G(1,i)$. We know from property $ii)$ that there is some $i'$ such that $d_G(1,i')=d_G(1,j)$ and $\{i',i\}$ is an edge in both graphs. Therefore if the initial segment of a path includes the edge $\{i,j\}$, then the shortest this can be is $d_G(1,i)$. If the initial segment ends at $j$ rather than $i$ then we know we have a shorter path to $j$ that we can replace this initial segment by. Otherwise it ends at $i$. We know that we have a path of length $d_G(1,j)$ to $i'$ on which we can attach the edge $\{i,i'\}$ to construct a new path of the same length but not using this edge.

We now have that condition $i)$ for Penrose active edges is the same in both graphs, since the graph distance is the same. We now indicate that an edge satisfies condition $ii)$ independent of the presence of $\{i,j\}$. We realise if $\{i,j\}$ was added or removed satisfying $i)$ then it has no effect on an edge satisfying $ii)$, since $ii)$ depends on edges between generations. Therefore, we consider that $\{i,j\}$ satisfies $ii)$. Since $d_G(1,i)=d_G(1,j)+1$, we have an $i'<j$ such that $\{i',i\}$ is an edge in both graphs and $d_G(1,i')=d_G(1,j)$. This means that if we use $j$ to invoke applying condition $ii)$ for an edge to be Penrose active, then we can invoke it in both cases by using $i'$. 
\end{proof}

We can also go the other way and find a Bernardi construction to provide an appropriate partition. The map $R:\mathfrak{a}[n] \to \mathcal{C}[n]$, which adds to $\tau$ all externally active edges for the given tree graph $\tau$ is the appropriate partition scheme. This is explained in the context of matroids below.

In the work of Bj\"{o}rner and Sokal \cite{B92, S05}, it is explained that for a matroid $M$, where we give a total order to the underlying set $E(M)$, we may find a partition of the collection of subsets of $E=E(M)$ according to the matroid structure. We introduce below some key definitions for matroids to introduce this connection, which can be found in the book of Oxley \cite{O92} and the work of Faris \cite{F12}.

A matroid $M$ on the ground set $E(M)=E$ is defined by a collection of independent subsets, denoted $\mathcal{I}(M)=\mathcal{I}$. These subsets must satisfy the following three axioms:

\begin{enumerate}

\item $\emptyset \in \mathcal{I} $ (non empty)
\item If $X \in \mathcal{I}$ and $X' \subset X$ then $X' \in \mathcal{I}$ (downward closed)
\item If $X \in \mathcal{I}$ and $Y \in \mathcal{I}$ and $|X| < |Y|$, then there exists $l \in Y\setminus X$ with $X \cup \{l\} \in \mathcal{I}$ (augmentation property)

\end{enumerate}
For a graphical matroid, the ground set is $[n]^{(2)}$. We define the independent sets as forests or acyclic graphs. 
\begin{definition}
A maximal independent set $X \in \mathcal{I}(\mathcal{M})$ is called a basis. The set of bases is denoted $\mathcal{B}(M)$.
\end{definition}
The maximal independent sets for a graphical matroid are therefore trees. 
\begin{definition}
The rank of a matroid $M$, $\mathrm{rk}(M)$ is the cardinality of a basis element.
\end{definition}
All bases have the same cardinality and so the rank is well defined. A matroid can be defined by its set of bases, since $X \in \mathcal{I}(M)$ if and only if $X \subseteq Y$, for some $Y \in \mathcal{B}(M)$. 

\begin{definition}[Restricted Matroid and Rank] Given a matroid, $M$, consider $X \subseteq E(M)$. There is a matroid $M|_X$, which is the restriction of $M$ to $X$. It has ground set $X$ and $\mathcal{I}(M|_X)=\{Y \in \mathcal{I}(M) | Y \subseteq X \}$.

For $X \subseteq E(M)$, the rank of $X$, $\mathrm{rk}(X)$ is the rank of the matroid $M|_X$ or alternatively the cardinality of the largest independent subset of $X$.
\end{definition}
We note that $\mathrm{rk}(X)=|X|$ if and only if $X$ is independent, so the rank function completely determines the matroid.

\begin{definition}[Dual of a Matroid]
The dual of a matroid is defined on the same ground set, but has a dual rank function rk$^{\star}$, defined by:
\be \text{rk}^{\star}(A):=|A|-\text{rk}(E)+\text{rk}(E \setminus A) \ee
\end{definition}

Let $\mathcal{B}$ be the set of bases for $E$. The dual basis set is then $\mathcal{B}^{\star}=\{E\setminus B \vert B \in \mathcal{B}\}$. We fix a total order on $E$ in the following.

\begin{definition}[Externally Active]
Let $B \in \mathcal{B}$. An element $e \in E\setminus B$ is \emph{externally active} on $B$ if $e$ is dependent on the list of elements of $B$ larger than it. We let $\tilde{\varepsilon}(B)$ be the set of externally active elements.
\end{definition}

\begin{definition}[Internally Active]
An element $e \in B$ is \emph{internally active} on $B$, if in the dual matroid $e$ is externally active on the complement $B^c=E \setminus B \in \mathcal{B}^{\star}$.  We denote by $\tilde{\iota}(B)=\tilde{\varepsilon}^{\star}(B^c)$ the set of internally active elements.
\end{definition}

For $R \subseteq S \subseteq E$, we define $[R,S]=\{A \vert R \subseteq A \subseteq S\}$. 

\begin{proposition}
$2^E$ can be written as the disjoint union:
\be 2^E=\bigsqcup_{B \in \mathcal{B}}[B\setminus \tilde{\iota}(B),B \cup \tilde{\varepsilon}(B)] \label{eq:C1} . \ee
\end{proposition}

For the case of the graphical matroid, we recall that the bases are the collection of trees. If we use the lexicographical order on the edges, then an edge is externally active for a tree $\tau$ in this sense, if and only if it is externally active in the sense of Bernardi \cite{B08}. This is due to the fact that all independent sets are forests and so a set of edges is dependent if it creates a cycle. We emphasise that for connected graphs, internally active edges play no role, since trees are minimally connected graphs. This therefore gives, when we intersect each set with connected graphs:

\be \mathcal{C}[n]=\bigsqcup_{\tau \in \mathfrak{a}[n]}[\tau, R(\tau)] ,\ee 
where $R(\tau)$ has edge set $E(\tau) \cup \tilde{\varepsilon}(\tau)$.

We note that the Penrose construction does not fit in the construction given above. In Figure \ref{fig:penrosenoorder}, we see that we would add the dashed edge in each case. In order to do this, we cannot have a consistent ordering on the edges $\{2,3\}$, $\{2.4\}$ and $\{3,4\}$. 

\begin{figure}[H]
\centering
\includegraphics[scale=0.2]{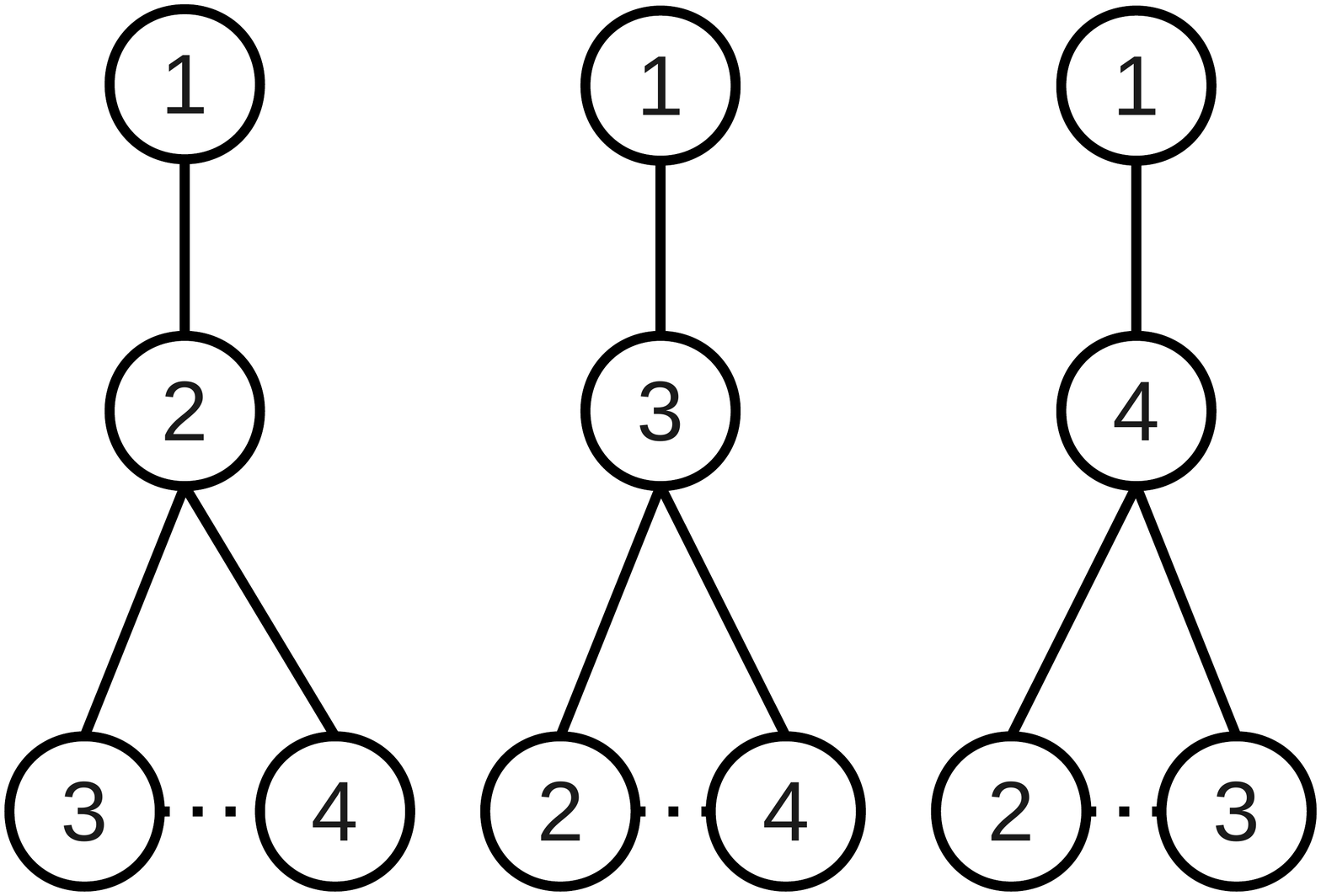}
\caption{Three Graphs to Indicate that the Penrose Construction is Different}
\label{fig:penrosenoorder}
\end{figure} 

The motivation of emphasising this connection is to understand if a similar connection may be drawn for two-connected graphs as the important context of the result.

\section{Outlook and Conclusions}
\label{sec:conclusion}
The main conclusion is that we are able to identify combinatorially the cancellations in the alternating sums of weighted two-connected graphs. The combinatorial factor arises from increasing trees on the subset $[n-1]$ of the vertices, with the vertex $n$ adjacent to every other vertex. There are modified versions for this in the case of the polytope, where we have the isomorphic graph structures, differing only through a relabelling in the form $i \mapsto i+s \text{ (mod) } n$ for all $s \in [n]$.

The key outlook for the work contained in this paper is to modify the set up explained in section \ref{sec:extend} towards two-connected graphs so that we obtain a helpful resummation of the graphs amenable to suitable estimation, which is important for the virial expansion. The parallel that is useful to draw here is that for the cluster expansion, we have the increasing and Cayley trees as the combinatorial objects representing the two cases above. It has been shown by Groeneveld \cite{G62} that these examples provide the extreme cases for positive potentials and an adaptation is available for stable potentials.

\medskip\noindent
{\bf Acknowledgements.}
The author would like to thank the anonymous referees for constructive comments in improving this article. The work for this paper has been funded by EPSRC grant EP/G056390/1 and SFB TR12. The author acknowledges helpful discussions with D. Brydges, R. Koteck\'{y}, D. Ueltschi and S. Jansen for particular discussions relating to this paper.

\end{document}